\documentclass{eptcs}
\providecommand{\event}{COS 2013} 
\usepackage{breakurl}             
\usepackage{graphicx}
\usepackage{amsmath}
\usepackage{amssymb}
\usepackage{amsthm}
\usepackage{stmaryrd}
\usepackage{array}
\usepackage{bussproofs}
\usepackage{float}
\floatstyle{boxed}
\restylefloat{figure}

\newcommand{\axc}[1]{\AxiomC{#1}}
\newcommand{\uic}[2]{\RightLabel{\small{#2}}\UnaryInfC{#1}}
\newcommand{\bic}[2]{\RightLabel{\small{#2}}\BinaryInfC{#1}}
\newcommand{\tic}[2]{\RightLabel{\small{#2}}\TrinaryInfC{#1}}

\newcommand{\hyp}{\mathsf{hyp}}
\newcommand{\wkn}[1]{\mathsf{wkn}{(#1)}}
\newcommand{\inl}[1]{\mathsf{inl}{(#1)}}
\newcommand{\inr}[1]{\mathsf{inr}{(#1)}}
\newcommand{\lam}[1]{\mathsf{lam}{(#1)}}
\newcommand{\reset}[1]{\mathsf{reset}{(#1)}}
\newcommand{\shift}[1]{\mathsf{shift}{(#1)}}
\newcommand{\casemy}[3]{\mathsf{case}({#1},{#2},{#3})}
\newcommand{\app}[2]{\mathsf{app}({#1},{#2})}

\DeclareMathOperator{\Nil}{nil}
\DeclareMathOperator{\Cons}{cons}
\DeclareMathOperator{\Unit}{Unit}
\DeclareMathOperator{\one}{tt}

\DeclareMathOperator{\Type}{Type}
\DeclareMathOperator{\Bool}{Bool}
\DeclareMathOperator{\Formula}{Formula}
\DeclareMathOperator{\List}{List}
\DeclareMathOperator{\Syn}{Syntax}
\DeclareMathOperator{\Val}{Value}
\DeclareMathOperator{\Ans}{Answer}
\DeclareMathOperator{\Fst}{fst}
\DeclareMathOperator{\Snd}{snd}
\DeclareMathOperator{\FV}{FV}
\newcommand{\normal}{{\!\!\text{nf}}}
\newcommand{\neutral}{{\!\!\text{ne}}}
\newcommand{\forces}[3]{{#1}\Vdash_{#2}{#3}}
\newcommand{\sforces}[3]{{#1}\Vdash^{\text{s}}_{#2}{#3}}
\newcommand{\run}[1]{\text{run}{(#1)}}
\newcommand{\ret}[1]{\text{return}{(#1)}}
\newcommand{\bind}[2]{\text{bind}{(#1,#2)}}
\newcommand{\nil}{{\Nil}}
\newcommand{\cons}[2]{{\Cons{(#1,#2)}}}
\newcommand{\Universal}{\mathcal{U}}

\newcommand{\lsup}[1]{{{^#1}}\!\!\!}
\newcommand{\reify}[4]{\lsup{#1}\downarrow_{#2}^{\!{#3}}({#4})}
\newcommand{\reflect}[4]{\lsup{#1}\uparrow_{#2}^{\!{#3}}({#4})}
\newcommand{\eval}[2]{\llbracket{#1}\rrbracket_{#2}}
\newcommand{\Gammareflect}[2]{\text{reflect}({#1},{#2})}
\newcommand{\ureset}[1]{\langle{#1}\rangle}
\newcommand{\ushift}[2]{\mathcal{S}{#1}.{#2}}

\theoremstyle{definition}
\newtheorem{definition}{Definition}[section]
\newtheorem*{definition*}{Definition}
\theoremstyle{plain}
\newtheorem{proposition}[definition]{Proposition}
\newtheorem{lemma}[definition]{Lemma}
\newtheorem{theorem}[definition]{Theorem}
\newtheorem{corollary}[definition]{Corollary}
\newtheorem*{proposition*}{Proposition}
\newtheorem*{lemma*}{Lemma}
\newtheorem*{theorem*}{Theorem}
\newtheorem*{corollary*}{Corollary}
\theoremstyle{remark}
\newtheorem{remark}[definition]{Remark}

\newtheorem*{remark*}{Remark}

\title{Type Directed Partial Evaluation for Level-1 Shift and Reset} 
\author{Danko Ilik\thanks{This work is covered by a Kurt G\"odel Research Prize Fellowship 2011}
  \institute{Laboratory for Complex Systems and Networks\\Macedonian Academy of Sciences and Arts\\Skopje, Macedonia}
  \email{danko.ilik@gmail.com}
}

\def\titlerunning{Type Directed Partial Evaluation for Level-1 Shift and Reset}
\def\authorrunning{Danko Ilik}
\begin{document}
\maketitle

\begin{abstract} 
  We present an implementation in the Coq proof assistant of type
  directed partial evaluation (TDPE) algorithms for call-by-name and
  call-by-value versions of shift and reset delimited control
  operators, and in presence of strong sum types. We prove that the
  algorithm transforms well-typed programs to ones in normal
  form. These normal forms can not always be arrived at using the so
  far known equational theories. The typing system does not allow
  answer-type modification for function types and allows delimiters to
  be set on at most one atomic type. The semantic domain for
  evaluation is expressed in Constructive Type Theory as a dependently
  typed monadic structure combining Kripke models and continuation
  passing style translations.  
\end{abstract}

\section{Introduction}

Type directed partial evaluation (TDPE) is a technique that partially evaluates a program by first compiling it, and pre-computing known (``static'') input data on the fly, and then decompiling it to normal form in an efficient process driven by the program's type. It was discovered by Danvy \cite{Danvy1999} in Programming Languages Theory, although the exact same algorithm had been isolated at about the same time also in the study of typed lambda calculi and in Logic: Berger and Schwichtenberg \cite{BergerS1991} found it while looking for an efficient procedure for reducing \emph{open} lambda terms and called it Normalization by Evaluation (NBE); Catarina Coquand \cite{CCoquand1993} realized that it is the procedure behind the proof of completeness of minimal intuitionistic logic (without $\bot, \vee$ and $\exists$) with respect to Kripke models.

However, when one moves from simply typed lambda calculus towards richer programming languages, to extend the TDPE method to cope with the new constructs does not appear to be straightforward. Already adding strong sum types seems to require one to implement TDPE using delimited control operators -- indeed, this is one of the more important applications of Danvy and Filinski's operators shift and reset \cite{DanvyF1990}. In turn, when considering TDPE for a language extended with the delimited control operators themselves, there has only been preliminary work on the subject, for the call-by-value case, by Tsushima and Asai \cite{TsushimaAsai2009}.

In this paper, we consider TDPE for the first level of the shift and reset hierarchy. Using their simpler non-extended CPS semantics, we build a type-theoretic framework that acts as a specification for TDPE algorithms (Section~\ref{sec:model}). The algorithms themselves, for both call-by-value and call-by-name, are given in Section~\ref{sec:algorithm}, where we also look at specific examples and compare their partial evaluations to the ones predicted by the known equational theories. In the concluding Section~\ref{sec:conclusion}, we give further explanation about our implementation and about the related works.

The Coq implementation of the algorithms can be found at the address \href{http://dankoi.github.com/metamath/}{dankoi.github.com/metamath}. Originally, this work was conceived as an alternative normalization proof for the core logical system from \cite{Ilik2010}, a proper constructive extension of intuitionistic logic with delimited control operators.

\section{Type-theoretic Model}\label{sec:model}

The programming language that we want to partially evaluate, our \emph{object language}, will be the lambda calculus with function and sum types and the shift/reset delimited control operators, described in Table~\ref{tab:typing}. We do not work with the most general known typing system for shift and reset in which implication is a quaternary connective \cite{DanvyF1989} and we allow a delimiter ($\reset{\cdot}$) to be set only at an atomic type ($\bot$).

\begin{table}
  \centering
  \begin{tabular}{ m{7cm} m{7cm} }
    \begin{prooftree}
      \axc{$~$}
      \uic{$\hyp : A,\Gamma\vdash_b A$}{~}
    \end{prooftree}
    &
    \begin{prooftree}
      \axc{$p : \Gamma\vdash_b A$}
      \uic{$\wkn{p} : B,\Gamma\vdash_b A$}{~}
    \end{prooftree}
    \\
    \begin{prooftree}
      \axc{$p : \Gamma\vdash_b A$}
      \uic{$\inl{p} : \Gamma\vdash_b A\vee B$}{~}
    \end{prooftree}
    &
    \begin{prooftree}
      \axc{$p : \Gamma\vdash_b B$}
      \uic{$\inr{p} : \Gamma\vdash_b A\vee B$}{~}
    \end{prooftree}
    \\
    \multicolumn{2}{ m{14cm} }{
      \begin{prooftree}
        \axc{$p : \Gamma\vdash_b A\vee B$}
        \axc{$q : A,\Gamma\vdash_b C$}
        \axc{$r : B,\Gamma\vdash_b C$}
        \tic{$\casemy{p}{q}{r} : \Gamma\vdash_b C$}{~}
      \end{prooftree}
    }
    \\
    \begin{prooftree}
      \axc{$p : A,\Gamma\vdash_b B$}
      \uic{$\lam{p} : \Gamma\vdash_b A\to B$}{~}
    \end{prooftree}
    &
    \begin{prooftree}
      \axc{$p : \Gamma\vdash_b A\to B$}
      \axc{$q : \Gamma\vdash_b A$}
      \bic{$\app{p}{q} : \Gamma\vdash_b B$}{~}
    \end{prooftree}
    \\
    \begin{prooftree}
      \axc{$p : \Gamma\vdash_1 \bot$}
      \uic{$\reset{p} : \Gamma\vdash_b \bot$}{~}
    \end{prooftree}
    &
    \begin{prooftree}
      \axc{$p : A\to \bot,\Gamma\vdash_1 \bot$}
      \uic{$\shift{p} : \Gamma\vdash_1 A$}{~}
    \end{prooftree}
  \end{tabular}
  
  \caption{A typing system for lambda calculus with sum types and shift and reset, where variable binding is handled using deBruijn indices ($\hyp$ and $\wkn{\cdot}$)}
  \label{tab:typing}
\end{table}

For expressing variable binding, we rely on deBruijn indices in the form of $\hyp$ and $\wkn{\cdot}$ rules, where $\hyp$ can be thought of as zero and $\wkn{\cdot}$ as the successor. Lambda abstraction ($\lam{\cdot}$) and control ($\shift{\cdot}$) are therefore unary.

The turnstile ``$\vdash$'' is annotated by a Boolean $b$ of value 0 or 1, value 1 meaning that a delimiting $\reset{}$ has been previously applied in the lambda term (typing tree derivation). All rules, except for $\shift{\cdot}$ and $\reset{\cdot}$ ignore this annotation $b$. The rule $\reset{}$ sets it to 1, and $\shift{}$ can only be used if the annotation has been previously set to 1 i.e. in a delimited sub-term.

The idea behind every TDPE algorithm is the following: we want to transform a program written in the object language to a meta-level ``bytecode'' version of it, ``run'' this bytecode  (this is called the \emph{evaluation} phase), and then, based on the program's type, recover a program in the object language that is already in normal form (the \emph{reification} phase) and $\beta(\eta)$-equal to the starting one. In other words, one relies on normalization at the meta-level, to produce an object-level normal form. This becomes non-trivial if the meta-level and the object level are not essentially the same, like in our case where the meta-level has no control feature while the object-level does.

The essential choice to make is what to choose for the ``semantic'' meta-level structure that evaluation will take place in. CPS semantics imposes itself, because it is the orthodox and simplest way to specify shift and reset  \cite{DanvyF1992}. The TDPE will thus be of the form of the two-phase transformation,
\[
\Syn \rightsquigarrow \left(\left(\Val\to\Ans\right)\to\Ans\right) \rightsquigarrow \Syn,
\]
where $\Syn$ denotes the type of programs of the object language, and $\Val$ and $\Ans$ are ``values'' and ``answers'' of a ``continuation'' in the usual terminology \cite{Danvy1999}. If we also want the transformation to account for \emph{open} terms (which allows to do normalization below a binder) and to guarantee that the input and output programs are actually programs of the same type, we need to enrich the  semantic domain (bytecode) by a pre-order, keeping track of a context $\Gamma$ denoting open variables, and a parameter $A$ corresponding to the type of the transformed program. 
We obtain the statement
\[
\Gamma\vdash A \rightsquigarrow \Gamma \Vdash A \rightsquigarrow \Gamma\vdash^\normal A,
\]
where $\Gamma\Vdash A$ denotes the semantic domain that is the target of evaluation (Subsection~\ref{subsec:soundness}), and source of reification (Subsection~\ref{subsec:completeness}). 

\subsection{Evaluating into the Models}\label{subsec:soundness}

We will use a combination of Kripke-CPS models for classical logic (used previously with Lee and Herbelin for proving NBE for the classical sequent calculus LK$_{\mu\tilde\mu}$ \cite{IlikLH2010}) and those for intuitionistic logic (used for proving NBE for intuitionistic natural deduction with $\vee$ and $\exists$ \cite{Ilik2011}). We give the mathematical definitions, trying to be precise but as informal as possible -- the interested reader may find the fully formal version in the Coq implementation -- keeping also in mind that, while we do use dependent types, the dependencies are rather weak ($\Pi$-types over the small set of formulas, booleans, and the type $K$).

\begin{definition}A \emph{Kripke CPS structure} is given by a type $K$, a relation $\le : (K \to K \to \Type)$ that is a preorder, i.e. both of
\begin{align*}
w\le w & &\text{(reflexivity)}\\
w_1\le w_2 \to w_2\le w_3 \to w_1\le w_3 & &\text{(transitivity)}
\end{align*}
hold, and a relation
\[
X : K \to \Bool \to \Formula \to \Type
\]
with the properties:
\begin{align*}
w_1\le w_2 &\to X w_1 b A \to X w_2 b A &(\le\text{-monotonicity})\\
b_1\sqsubseteq b_2 &\to X w b_1 A \to X w b_2 A &(\sqsubseteq\text{-monotonicity})\\
X w 1 \bot &\to X w 0 \bot &(\text{meta-}\reset{}).
\end{align*}

$\Bool$ is the type of booleans with inhabitants 0 (false) and 1 (true), and $\sqsubseteq$ is the order on booleans defined by the relation of less-than-or-equal of their numerical values. ``$\Type$'' denotes the type universe of the meta-language, while $\Formula$ is the type of types of the object language i.e. those built from $\to$, $\vee$, atomic types, and the special fixed type $\bot$ that reset can be set on -- we do not make the usual assumption that $\bot$ denotes the empty type, it is simply a notation for a chosen atomic type.

Inhabitants $w$ of type $K$ are called \emph{worlds}, and when we have $X w b A$ we say that the world $w$ is \emph{exploding} for the formula $A$ with annotation $b$. This terminology (``exploding'' or ``fallible'') comes from classic use of Kripke models when interpreting absurdity in a constructive way \cite{TroelstraVD1}. The relation $X$, the answer type of the continuations, will later be instantiated with the set of typable terms in normal form, that is, it will be used to pass on the output of the TDPE between different sub-phases of the algorithm in the process of building the final normal form.
\end{definition}

\begin{definition}\label{def:forcing}Given $F : (K\to\Bool\to\Formula\to\Type)$, $A:\Formula$, $b:\Bool$, $w:K$, the dependently typed continuations ``monad'', $\forces{w}{b}{A}$, defined by
\begin{multline*}
\forces{w}{0}{A} := (C:\Formula)(w_1:K)(w\le w_1 \to \\(w_2:K)(w_1\le w_2 \to F w_2 0 A \to X w_2 0 C) \to X w_1 0 C)
\end{multline*}
\begin{multline*}
\forces{w}{1}{A} := (w_1:K)(w\le w_1 \to \\(w_2:K)(w_1\le w_2 \to F w_2 1 A \to X w_2 1 \bot) \to X w_1 1 \bot)
\end{multline*}
is called \emph{forcing}. That is, we read $\forces{w}{b}{A}$ as ``the world $w$ forces the type $A$ with annotation $b$''. 
\end{definition}

\begin{remark}
  We have put the word ``monad'' in quotes because we have not sought to prove the usual categorical or the functional programming laws for monads hold. Yet, the fact that we can define the monadic unit, bind, and run, will be quite convenient for structuring the computation/proofs later on.
\end{remark}

The following two definitions present two alternatives that can be used to instantiate $F$ from Definition~\ref{def:forcing}; when the (non-strong) forcing relation is used in the definitions, it is implicitly instantiated with the strong forcing relation being defined. Note that, type theoretically, (non-strong) forcing and strong forcing need not be defined simultaneously, Definition~\ref{def:forcing} comes first.

\begin{definition}[Strong forcing, call-by-value variant] The \emph{strong forcing} relation $\sforces{w}{b}{A}$ is defined by recursion on the type $A$, by the following clauses:
\begin{align*}
\sforces{w}{b}{A} &:= X w b A \qquad (A -\text{atomic type})\\
\sforces{w}{b}{A\vee B} &:= \sforces{w}{b}{A} + \sforces{w}{b}{B}\\
\sforces{w}{b}{A\to B} &:= (w':K)(w\le w' \to \sforces{w'}{b}{A} \to \forces{w'}{b}{B})
\end{align*}
\end{definition}

\begin{definition}[Strong forcing, call-by-name variant] The \emph{strong forcing} relation $\sforces{w}{b}{A}$ is defined by recursion on the type $A$, by the following clauses:
\begin{align*}
\sforces{w}{b}{A} &:= X w b A \qquad (A -\text{atomic type})\\
\sforces{w}{b}{A\vee B} &:= \forces{w}{b}{A} + \forces{w}{b}{B}\\
\sforces{w}{b}{A\to B} &:= (w':K)(w\le w' \to \forces{w'}{b}{A} \to \forces{w'}{b}{B})
\end{align*}
\end{definition}

Although a different strong forcing relation $(\sforces{\cdot}{\cdot}{\cdot})$ determines a different forcing relation $(\forces{\cdot}{\cdot}{\cdot})$, the important properties that hold of the latter are nonetheless the same regardless of which strong forcing was chosen.

\begin{lemma} The following properties hold of strong and ordinary forcing:
  \begin{align*}
    w\le w' &\to \sforces{w}{b}{A} \to \sforces{w'}{b}{A}\\
    w\le w' &\to \forces{w}{b}{A} \to \forces{w'}{b}{A}\\
    b\sqsubseteq b' &\to \sforces{w}{b}{A} \to \sforces{w}{b'}{A}\\
    b\sqsubseteq b' &\to \forces{w}{b}{A} \to \forces{w}{b'}{A}\\
    \forces{w}{b}{\bot} &\to X w b \bot & (\run{\cdot})\\
    \sforces{w}{b}{A} &\to \forces{w}{b}{A} & (\ret{\cdot})\\
    (w':K)(w\le w' &\to \sforces{w'}{b}{A}\to\forces{w'}{b}{B}) \\
    & \to \forces{w}{b}{A}\to \forces{w}{b}{B} & (\bind{\cdot}{\cdot})
  \end{align*}
\end{lemma}
\begin{proof} The proofs of monotonicity of strong forcing with respect to $\le$ and $\sqsubseteq$ are done by induction on the formula $A$ using monotonicity of $X$. Monotonicity of (non-strong) forcing requires no induction. The proofs of $\run{\cdot}, \ret{\cdot}$, and $\bind{\cdot}{\cdot}$ follow the structure given on Figure~\ref{fig:glue}.
\end{proof}

We will use the same turnstile symbols to denote forcing and strong forcing of \emph{finite lists} of formulas, $\Gamma$, defined by,
\begin{align*}
  \forces{w}{b}{\nil} &:= \Unit\\
  \forces{w}{b}{\cons{A}{\Gamma}} &:= \forces{w}{b}{A}\times \forces{w}{b}{\Gamma} \\
  \sforces{w}{b}{\nil} &:= \Unit\\
  \sforces{w}{b}{\cons{A}{\Gamma}} &:= \sforces{w}{b}{A}\times \sforces{w}{b}{\Gamma},
\end{align*}
where $\times$ is the product type (i.e. logical conjunction, when used as a predicate) and $\Unit$ is the singleton type. Naturally, the monotonicity properties from the previous lemma extend to forcing and strong forcing for lists.

\begin{theorem}[Evaluation for call-by-name]\label{thm:eval:cbn} If $p: \Gamma\vdash_b A$, then for any $w$ and any $b'$ such that $b\sqsubseteq b'$ we have that from the finite product $\forces{w}{b'}{\Gamma}$ we can construct $\forces{w}{b'}{A}$.
\end{theorem}
\begin{theorem}[Evaluation for call-by-value]\label{thm:eval:cbv} If $p: \Gamma\vdash_b A$, then for any $w$ and any $b'$ such that $b\sqsubseteq b'$ we have that from the finite product $\sforces{w}{b'}{\Gamma}$ we can construct $\forces{w}{b'}{A}$.
\end{theorem}
\begin{proof}The proofs of both theorems are done in continuation-passing style, by using induction on the derivation of $p$. The program skeletons that corresponds to the proofs can be seen on figures~\ref{fig:evalcbn} and~\ref{fig:evalcbv}, 
 and the full proofs are available in the Coq formalization.\end{proof}

\subsection{Reifying from the Models}\label{subsec:completeness}

While the evaluation theorems from the previous subsection can be used for any concrete structure that implements the Kripke-CPS models axiomatization, in this section we build one such model, $\Universal$, the \emph{universal model}, from syntactic elements. It gets its name from the fact that if something is forced in $\Universal$ then it is also forced in any other possible model.

To obtain a finer grained characterization of the TDPE procedure, we will separate the lambda terms into a level of \emph{normal terms} and a level of \emph{neutral terms} using the following inductive definition.
\begin{align*}
(-\vdash^\normal_b-) \ni r &::= \lam{r} ~|~ \inl{r} ~|~ \inr{r} ~|~ \shift{r} ~|~ e \\
(-\vdash^\neutral_b-) \ni e &::= \app{e}{r} ~|~ \casemy{e}{r_1}{r_2} ~|~ \reset{e} ~|~ \hyp ~|~ \wkn{r}
\end{align*}
This definition concerns \emph{typed} lambda terms (i.e. typing tree derivations), although typing information has been suppressed.

The separation into normal versus neutral terms is standard in the NBE literature, but what is new here is that, in order to obtain the Disjunction Property at the end of this section, 
$\reset{\cdot}$ has to be neutral.

\begin{definition}[The model $\Universal$] The universal Kripke-CPS model $\Universal$ is built when the set of worlds is the set of contexts $\Gamma$,
\[
K := \List(\Formula),
\]
and the predicate $X$ is defined by recursion on the structure of types of the object language,
\begin{align*}
  X \Gamma b A &:= \Gamma \vdash^\neutral_b A \qquad (A -\text{atomic type})\\
  X \Gamma b \bot &:= \Gamma \vdash^\neutral_b \bot\\
  X \Gamma b (A \vee B) &:= \Gamma \vdash^\normal_b A\vee B \\
  X \Gamma b (A \to B) &:= \Gamma \vdash^\normal_b A\to B,
\end{align*}
as the set of terms in normal or neutral form of the given type.

The pre-order $\le$ is defined as the prefix relation on lists. It is not hard to see that reflexivity and transitivity of $\le$ hold, and that $\le$-monotonicity and $\sqsubseteq$-monotonicity hold by the weakening properties of the typing system (formal lemmas \texttt{proof\_nf\_mon}, \texttt{proof\_ne\_mon}, \texttt{proof\_nf\_mon2}, and \texttt{proof\_ne\_mon2}). The property $\text{meta-}\reset{}$ is provided by the syntactic $\reset{}$ rule (formal lemma \texttt{X\_reset}).
\end{definition}

We can now prove that for any meta-level evaluation there exists a term in the object language (\emph{reification} part). Due to contravariance of implication (function types), we need a simultaneous map in the other direction (\emph{reflection} part) \footnote{Note that, while reflection and evaluation (theorems~\ref{thm:eval:cbn} and~\ref{thm:eval:cbv}), have the same typing, the first just does eta-expansions by recursion on the object-language type, while the latter is more informative being defined by recursion on the object-language \emph{term}. }.

\begin{theorem}[Reification ($\downarrow$) and reflection ($\uparrow$)] Given $A:\Formula$, $\Gamma:\List(\Formula)$ and $b:\Bool$, the following two statements hold:
  \begin{align}
    \forces{\Gamma}{b}{A} &\to \Gamma\vdash^\normal_b A & \text{``reify''} \tag{$\reify{\Gamma}{b}{A}{\cdot}$}\\
    \Gamma\vdash^\neutral_b A &\to \forces{\Gamma}{b}{A} & \text{``reflect''} \tag{$\reflect{\Gamma}{b}{A}{\cdot}$}
  \end{align}
\end{theorem}
\begin{proof} The two statements are proved simultaneously, by induction on the type $A$. The program skeleton corresponding to the proof can be seen on figures~\ref{fig:reifycbn} and~\ref{fig:reifycbv}. The full proof is done in continuation passing style and is available in the Coq formalization.
\end{proof}

Let $\Gammareflect{\Gamma}{b}$ denote the fold-left of the list $\Gamma$ for the reflection function applied to a variable ($\hyp$), using the unit type constructor $\one$ in the base case. For example, for $\Gamma := \cons{A}{\cons{B}{\cons{C}{\nil}}}$, we have
\begin{align*}
\Gammareflect{\Gamma}{b} &: \forces{\Gamma}{b}{A} \times \forces{\Gamma}{b}{B} \times \forces{\Gamma}{b}{C} \times \Unit\\
\Gammareflect{\Gamma}{b} &= \reflect{\Gamma}{b}{A}{\hyp} , \reflect{\Gamma}{b}{B}{\hyp} , \reflect{\Gamma}{b}{C}{\hyp} , \one
\end{align*}

We can now obtain the main result of the paper by composing the Evaluation theorems with the Reification theorem, all of which have constructive proofs. In other words, we take a term $p$, apply a meta-CPS translation $\llbracket\cdot\rrbracket$ on it, in an initial environment built from the context $\Gamma$ by the reflect function, and then reconstruct a term in normal form based on the type $A$ using the reification function $(\downarrow\cdot)$.

\begin{corollary}[TDPE for call-by-name] Given $p : \Gamma\vdash_b A$, we have that $\reify{\Gamma}{b}{A}{\eval{p}{\Gammareflect{\Gamma}{b}}} : \Gamma\vdash^\normal_b A$.
\end{corollary}

\begin{corollary}[TDPE for call-by-value] Given $p : \nil\vdash_b A$, we have that $\reify{\nil}{b}{A}{\eval{p}{\Unit}} : \nil\vdash^\normal_b A$.
\end{corollary}

\begin{remark}
The difference in formulation between the two corollaries is due to the fact that the Evaluation theorem for call-by-name (Theorem~\ref{thm:eval:cbn}) uses ordinary forcing for the context $\Gamma$, while the corresponding Theorem~\ref{thm:eval:cbv} for call-by-value uses strong forcing. TDPE for CBN can therefore be run on open terms directly, while for CBV we have to have a closed term as input, although TDPE for CBV does normalize below lambda abstractions.
\end{remark}

The following property shows that the calculus from Table~\ref{tab:typing} can be considered a constructive logical system, despite the fact that it contains control operators which are usually connected with classical logic. (Classical logic does not have this property)
\begin{proposition}[Disjunction Property] If $p : \nil\vdash_0 A\vee B$ then from $p$ one can get $p'$ such that either $p' : \nil\vdash_0 A$ or $p' : \nil\vdash_0 B$.
\end{proposition}
\begin{proof}
We can use TDPE to transform $p$ to a term in normal form $r : \nil\vdash^\normal_0 A\vee B$. Now, from the syntax of normal and neutral forms, one can see that the only possibilities for $r$ are that it is either a $\inl{r'}$ or a $\inr{r'}$ -- $r$ can not be any of the neutral forms because it does not have a free variable (the context is $\nil$) -- and $r$ cannot be a $\shift{\cdot}$ because of the annotation 0 on the turnstile.
\end{proof}

\section{Algorithm}\label{sec:algorithm}

In this section we show the algorithmic core of the TDPE procedure. While the exact program in a dependently typed language can be seen with all its gory details in the Coq formalization, our intention here is to give a human readable account of the procedure that we extracted by hand from the Coq formalization. This extraction consists in deleting the dependently typed information which is mostly connected to handling worlds (members of the preorder $K$) and the associated monotonicity proofs.

We will use two levels of lambda calculus: on one level we will have the ``dynamic'' lambda terms from Table~\ref{tab:typing}, and on the other ``static'' level we will use ordinary mathematical function notation: ``$\mapsto$'' for abstraction, ``$\cdot$'' for application, $\iota_1$ for injection-left, $\iota_2$ for injection-right, and the usual big-open-curly-bracket for definition by cases. Small Greek letters $\alpha,\beta,\gamma,\phi,\kappa$ are used for static variables; there are no explicit dynamic variables since we use deBruijn indices. The equality symbol ``:='' denotes definitional equality.

The monadic glue functions are defined on Figure~\ref{fig:glue}. Parameters corresponding to dependent types for world-handling have been left out (worlds are marked with bars ``$-$'').
\begin{figure*}
  \centering
  \begin{align*}
    \ret{\cdot} &: \sforces{-}{b}{A} \to \forces{-}{b}{A}\\
    \ret{\alpha} &:= \kappa\mapsto\kappa\cdot\alpha \\
    ~ & ~\\
    \bind{\cdot}{\cdot} &: (\sforces{-}{b}{A} \to \forces{-}{b}{B}) \to \forces{-}{b}{A} \to \forces{-}{b}{B}\\  
    \bind{\phi}{\alpha} &:= \kappa\mapsto\alpha\cdot(\alpha'\mapsto \phi\cdot\alpha'\cdot\kappa)\\  
    ~ & ~\\
    \run{\cdot} &: \forces{-}{b}{\bot} \to \sforces{-}{b}{\bot}\\
    \run{\alpha} &:= \alpha\cdot(\chi\mapsto\chi)
  \end{align*}
  \caption{Monadic glue functions}
  \label{fig:glue}
\end{figure*}

The evaluation algorithms corresponding to theorems~\ref{thm:eval:cbv} and~\ref{thm:eval:cbn} are given on figures~\ref{fig:evalcbn} and \ref{fig:evalcbv}.

\begin{figure*}
  \centering
  \begin{align*}
    \eval{p : \Gamma\vdash_b A}{\forces{w}{b}{\Gamma}} &: \forces{w}{b}{A}\\
    \eval{\hyp}{\rho} &:= \Fst(\rho)\\
    \eval{\wkn{p}}{\rho} &:= \eval{p}{\Snd(\rho)}\\
    \eval{\lam{p}}{\rho} &:= \ret{\alpha\mapsto\eval{p}{\alpha,\rho}}\\
    \eval{\app{p}{q}}{\rho} &:= \bind{\phi\mapsto\phi\cdot\eval{q}{\rho}}{\eval{p}{\rho}}\\
    \eval{\inl{p}}{\rho} &:= \ret{\iota_1\eval{p}{\rho}}\\
    \eval{\inr{p}}{\rho} &:= \ret{\iota_2\eval{p}{\rho}}\\
    \eval{\casemy{p}{q}{r}}{\rho} &:= \bind{\gamma\mapsto \left\{ 
      \begin{array}{ll}
        \eval{q}{\alpha,\rho} & ,\text{ if } \gamma=\iota_1\alpha\\
        \eval{r}{\beta,\rho} & ,\text{ if } \gamma=\iota_2\beta
      \end{array}        
      \right.}{\eval{p}{\rho}}\\
    \eval{\shift{p}}{\rho} &:= \kappa\mapsto\run{\eval{p}{\ret{\alpha\mapsto\ret{\alpha\cdot\kappa}},\rho}}\\
    \eval{\reset{p}}{\rho}^{b=1} &:=\ret{\run{\eval{p}{\rho}}}\\
    \eval{\reset{p}}{\rho}^{b=0} &:=\ret{\text{meta-}\reset{\run{\eval{p}{\rho}}}}
  \end{align*}
  \caption{Evaluation for call-by-name}
  \label{fig:evalcbn}
\end{figure*}

\begin{figure*}
  \centering
  \begin{align*}
    \eval{p : \Gamma\vdash_b A}{\sforces{w}{b}{\Gamma}} &: \forces{w}{b}{A}\\
    \eval{\hyp}{\rho} &:= \ret{\Fst(\rho)}\\
    \eval{\wkn{p}}{\rho} &:= \eval{p}{\Snd(\rho)}\\
    \eval{\lam{p}}{\rho} &:= \ret{\alpha\mapsto\eval{p}{\alpha,\rho}}\\
    \eval{\app{p}{q}}{\rho} &:= \bind{\phi\mapsto\bind{\phi}{\eval{q}{\rho}}}{\eval{p}{\rho}}\\
    \eval{\inl{p}}{\rho} &:= \bind{\alpha\mapsto\ret{\iota_1\alpha}}{\eval{p}{\rho}}\\
    \eval{\inr{p}}{\rho} &:= \bind{\alpha\mapsto\ret{\iota_2\alpha}}{\eval{p}{\rho}}\\
    \eval{\casemy{p}{q}{r}}{\rho} &:= \bind{\gamma\mapsto \left\{ 
      \begin{array}{ll}
        \eval{q}{\alpha,\rho} & ,\text{ if } \gamma=\iota_1\alpha\\
        \eval{r}{\beta,\rho} & ,\text{ if } \gamma=\iota_2\beta
      \end{array}        
      \right.}{\eval{p}{\rho}}\\
    \eval{\shift{p}}{\rho} &:= \kappa\mapsto\run{\eval{p}{{\alpha\mapsto\ret{\kappa\cdot\alpha}},\rho}}\\
    \eval{\reset{p}}{\rho}^{b=1} &:=\ret{\run{\eval{p}{\rho}}}\\
    \eval{\reset{p}}{\rho}^{b=0} &:=\ret{\text{meta-}\reset{\run{\eval{p}{\rho}}}}
  \end{align*}  
  \caption{Evaluation for call-by-value}
  \label{fig:evalcbv}
\end{figure*}

The reification algorithms are defined by mutual recursion with reflection algorithms on figures~\ref{fig:reifycbn} and~\ref{fig:reifycbv}. For facilitating comparison, the places where call-by-value and call-by-name versions differ are marked with boxes.

\begin{figure*}
  \centering
  \begin{align*}
    \reify{\Gamma}{b}{A}{\cdot} &: \forces{\Gamma}{b}{A} \to \Gamma\vdash^\normal_b A\\
    \reify{\Gamma}{b}{\bot}{\alpha} &:= \run{\alpha}\\
    \reify{\Gamma}{0}{A_0}{\alpha} &:= \run{\alpha}\quad\quad\quad\qquad\qquad\qquad\qquad\qquad\qquad\qquad\qquad\qquad\qquad \text{for atomic }A_0\neq \bot\\
    \reify{\Gamma}{1}{A_0}{\alpha} &:= \shift{\alpha\cdot(\chi\mapsto\app{\hyp}{\chi})}\qquad\qquad\qquad\qquad\qquad\qquad\qquad \text{for atomic }A_0\neq \bot\\
    \reify{\Gamma}{b}{A\to B}{\alpha} &:= \lam{\reify{\Gamma}{b}{B}{\kappa\mapsto \reflect{{A,\Gamma}}{b}{A}{\hyp}\cdot (\alpha'\mapsto\alpha\cdot(\phi\mapsto\phi\cdot(\ret{\alpha'})\cdot\kappa))}}\\
    \reify{\Gamma}{0}{A\vee B}{\alpha} &:= \alpha\cdot\left(\gamma\mapsto\left\{
      \begin{array}{ll}
        \inl{\reify{{\Gamma}}{0}{A}{\beta}} & ,\text{ if } \gamma=\iota_1\beta\\
        \inr{\reify{{\Gamma}}{0}{B}{\beta}} & ,\text{ if } \gamma=\iota_2\beta
      \end{array}
    \right.\right)\\
    \reify{\Gamma}{1}{A\vee B}{\alpha} &:=\shift{\alpha\cdot\left(\gamma\mapsto\left\{
      \begin{array}{ll}
        \app{\hyp}{\inl{\reify{{\Gamma}}{1}{A}{\beta}}} & ,\text{ if } \gamma=\iota_1\beta\\
        \app{\hyp}{\inr{\reify{{\Gamma}}{1}{B}{\beta}}} & ,\text{ if } \gamma=\iota_2\beta
      \end{array}
    \right.\right)}\\
    ~ & ~\\
    \reflect{\Gamma}{b}{A}{\cdot} &: \Gamma\vdash^\neutral_b A\to \forces{\Gamma}{b}{A}\\
    \reflect{\Gamma}{b}{A_0}{e} &:= \ret{e}\quad\quad\qquad\qquad\qquad\qquad\qquad\qquad\qquad\qquad\qquad\qquad \text{for atomic }A_0\\
    \reflect{\Gamma}{b}{A\to B}{e} &:= \ret{\alpha\mapsto\reflect{\Gamma}{b}{B}{\app{e}{\reify{\Gamma}{b}{A}{\alpha}}}}\\
    \reflect{\Gamma}{b}{A\vee B}{e} &:= \kappa\mapsto\casemy{e}{\reflect{{A,\Gamma}}{b}{A}{\hyp}\cdot(\alpha\mapsto\kappa\cdot\iota_1\ret{\alpha})}{\reflect{{B,\Gamma}}{b}{B}{\hyp}\cdot(\alpha\mapsto\kappa\cdot\iota_2\ret{\alpha})}
  \end{align*}
  \caption{Reification and reflection for call-by-name}
  \label{fig:reifycbn}
\end{figure*}

\begin{figure*}
  \centering
  \begin{align*}
    \reify{\Gamma}{b}{A}{\cdot} &: \forces{\Gamma}{b}{A} \to \Gamma\vdash^\normal_b A\\
    \reify{\Gamma}{b}{\bot}{\alpha} &:= \run{\alpha} \\
    \reify{\Gamma}{0}{A_0}{\alpha} &:= \run{\alpha}\quad\quad\quad\qquad\qquad\qquad\qquad\qquad\qquad\qquad\qquad\qquad\qquad \text{for atomic }A_0\neq \bot \\
    \reify{\Gamma}{1}{A_0}{\alpha} &:= \shift{\alpha\cdot(\chi\mapsto\app{\hyp}{\chi})}\qquad\qquad\qquad\qquad\qquad\qquad\qquad \text{for atomic }A_0\neq \bot\\
    \reify{\Gamma}{b}{A\to B}{\alpha} &:= \lam{\reify{\Gamma}{b}{B}{\kappa\mapsto \reflect{{A,\Gamma}}{b}{A}{\hyp}\cdot (\alpha'\mapsto\alpha\cdot(\phi\mapsto\phi\cdot(\boxed{\alpha'})\cdot\kappa))}}\\
    \reify{\Gamma}{0}{A\vee B}{\alpha} &:= \alpha\cdot\left(\gamma\mapsto\left\{
      \begin{array}{ll}
        \inl{\reify{{\Gamma}}{0}{A}{\boxed{\ret{\beta}}}} & ,\text{ if } \gamma=\iota_1\beta\\
        \inr{\reify{{\Gamma}}{0}{B}{\boxed{\ret{\beta}}}} & ,\text{ if } \gamma=\iota_2\beta
      \end{array}
    \right.\right)\\
    \reify{\Gamma}{1}{A\vee B}{\alpha} &:=\shift{\alpha\cdot\left(\gamma\mapsto\left\{
      \begin{array}{ll}
        \app{\hyp}{\inl{\reify{{\Gamma}}{1}{A}{\boxed{\ret{\beta}}}}} & ,\text{ if } \gamma=\iota_1\beta\\
        \app{\hyp}{\inr{\reify{{\Gamma}}{1}{B}{\boxed{\ret{\beta}}}}} & ,\text{ if } \gamma=\iota_2\beta
      \end{array}
    \right.\right)}\\
    ~ & ~\\
    \reflect{\Gamma}{b}{A}{\cdot} &: \Gamma\vdash^\neutral_b A\to \forces{\Gamma}{b}{A}\\
    \reflect{\Gamma}{b}{A_0}{e} &:= \ret{e}\quad\quad\qquad\qquad\qquad\qquad\qquad\qquad\qquad\qquad\qquad\qquad \text{for atomic }A_0\\
    \reflect{\Gamma}{b}{A\to B}{e} &:= \ret{\alpha\mapsto\reflect{\Gamma}{b}{B}{\app{e}{\reify{\Gamma}{b}{A}{\boxed{\ret{\alpha}}}}}}\\
    \reflect{\Gamma}{b}{A\vee B}{e} &:= \kappa\mapsto\casemy{e}{\reflect{{A,\Gamma}}{b}{A}{\hyp}\cdot(\alpha\mapsto\kappa\cdot\iota_1\boxed{\alpha})}{\reflect{{B,\Gamma}}{b}{B}{\hyp}\cdot(\alpha\mapsto\kappa\cdot\iota_2\boxed{\alpha})}
  \end{align*}  
  \caption{Reification and reflection for call-by-value}
  \label{fig:reifycbv}
\end{figure*}

\subsection{Known Equational Theories}

Before considering computational tests, we recall the available equational theories for shift and reset. 

The equational theory for call-by-value shift and reset, for the full hierarchy, has been proven sound and complete with respect to the extended CPS translation \cite{DanvyF1990} by Kameyama \cite{Kameyama2007}. Considering the first level of the hierarchy which is of interest here, the equations are expressed using the classes of \emph{values} ($V$) and \emph{pure evaluation contexts} ($F$),
\begin{align*}
  V &::= x ~|~ \lambda x.p &
  F &::= [] ~|~ F p ~|~ V F,
\end{align*}
as follows:
\begin{align}
(\lambda x.p) V &= p \{V/x\}\\
\lambda x. V x &= V & \text{ when }x\notin \FV(V)\\
(\lambda x. F[x])p &= F[p] & \text{ when }x\notin \FV(F)\\
\label{reset-value}\ureset{V} &= V\\
\ureset{(\lambda x.p)\ureset{q}} &= (\lambda x. \ureset{p})\ureset{q}\\
\ushift{k}{\ureset{p}} &= \ushift{k}{p}\\
\ushift{k}{k\ureset{p}} &= \ureset{p}&\text{ when }k\notin \FV(p)\\
\ushift{k}{k p} &= p&\text{ when }k\notin \FV(p)\\
\ureset{F[\ushift{k}{p}]} &= \ureset{p\{(\lambda x.\ureset{F[x]}) ~/~ k\}}&\text{ when }x\notin \FV(F)\cup\{k\}
\end{align}

The equational theory for call-by-name shift and reset, for the first level of the hierarchy, has been studied by Kameyama and Tanaka \cite{KameyamaTanaka2010}. For the purpose of proving soundness and completeness with respect to Biernacka and Biernacki's \cite{BiernackaB2009} call-by-name CPS semantics for shift and reset, Kameyama and Tanaka distinguish between two kinds of term applications, the usual one, and the one to continuation variables ($\hookleftarrow$); and two kinds of substitutions, for normal variables ($\{\cdot/x\}$), and for continuation variables ($\{k\Rightarrow\cdot\}$). The classes of values and pure evaluation contexts are restrictions of the call-by-name ones, given by:\footnote{Kameyama and Tanaka also consider constants $c$ among the call-by-name values, however no variables are allowed. Since we do not have constants in our minimal object-language of study, we did not include them as an option of $U$.}
\begin{align*}
  U &::= \lambda x.p &
  E &::= [] ~|~ E p
\end{align*}
The equational theory is as follows,
\begin{align}
(\lambda x.p) q &= p \{q/x\}\\
\ureset{U} &= U\\
k'\hookleftarrow E[\ushift{k}{p}] &= \ureset{p \{k\Rightarrow (k'\hookleftarrow E)\}}\\
\ushift{k}{\ureset{p}} &= \ushift{k}{p}\\
\ushift{k}{k\hookleftarrow p} &= p&\text{ when }k\notin \FV(p)\\
\ureset{E[\ushift{k}{p}]} &= \ureset{p\{k\Rightarrow E\}},
\end{align}
where the substitution $q \{k\Rightarrow\cdot\}$ is defined by recursive descent on the term $q$ and affects only the subterms of the form $k\hookleftarrow p$ by:
\begin{align*}
(k'\hookleftarrow p)\{k\Rightarrow E\} &= \ureset{E[p\{k\Rightarrow E\}]} &\text{when } k' = k\\
(k'\hookleftarrow p)\{k\Rightarrow E\} &= k'\hookleftarrow (p \{k\Rightarrow E\})&\text{when } k' \neq k
\end{align*}

We have used conventional syntax, writing $\reset{p}$ as $\ureset{p}$, and $\shift{p}$ as $\ushift{k}{p}$, and will continue to do so in the next subsection.

\subsection{Example Runs of the Algorithm}\label{examples}

Let us now consider some test-runs of our TDPE procedure. Each example consists of an input term, marked with a number to refer to, and two outputs: using TDPE for call-by-value (CBV) and for call-by-name (CBN).

We begin with simple examples where the continuation variable of shift is not used (exceptions effect).
\begin{align}
\lambda x.& \ureset{(\lambda y.y)(\ushift{k}{x})}\label{ex1}\\
\lambda x.& \ureset{x} \tag{CBV}\\
\lambda x.& \ureset{\ureset{x}} \tag{CBN}
\end{align}
\begin{align}
\lambda x.& \ureset{\ureset{\ureset{(\lambda y.y)(\ushift{k}{x})}}}\label{ex2}\\
\lambda x.& \ureset{x} \tag{CBV}\\
\lambda x.& \ureset{\ureset{x}} \tag{CBN}
\end{align}
The CBN normal forms are not perfect as the resets are systematically duplicated at top level. This duplication is not related to the number of reset as input, as can be seen from Example (\ref{ex2}), but to a ``bug'' in the Coq formalization. Namely, the lemma \texttt{Kont\_sforces\_mon2'}, proving monotonicity of non-strong forcing with respect to the Boolean order, uses a reset in the proof, and, since this lemma is not used in the CBV case, the problem does not appear there.\footnote{The solution might be to make the non-strong forcing monad monotone also for the $\sqsubseteq$ relation and not only for $\leq$ on worlds, and is the subject of future work.}

The CBV equational theory can derive the TDPE output for examples (\ref{ex1}) and (\ref{ex2}). The CBN equational theory derives $\lambda x.\ureset{x}$ but not $\lambda x.\ureset{\ureset{x}}$. However, our TDPE for CBN identifies the two, because it also normalizes $\lambda x.\ureset{x}$ to $\lambda x.\ureset{\ureset{x}}$.

The next example does not use a control operator, but has a delimiter.
\begin{align}
\lambda x y.& \ureset{\ureset{x y}}\label{ex3}\\
\lambda x y.& {\ureset{x y}} \tag{CBV}\\
\lambda x y.& \ureset{\ureset{x \ureset{y}}} \tag{CBN}
\end{align}
The CBV and CBN equational theories do not transform Example (\ref{ex3}) further, because the subterm $x y$ is not a value. TDPE for CBV removes one delimiter, as if $x y$ were a value, and TDPE for CBN delimits the inside variable y, as if it were taking into account that variables are not values according to the CBN equational theory.

Let us consider an example that uses the continuation inside a shift.
\begin{align}
\lambda x y.& \ureset{x (\ushift{k}{k(k y)})}\label{ex4}\\
\lambda x y.& {\ureset{x (x y)}} \tag{CBV}\\
\lambda x y.& \ureset{\ureset{x \ureset{y}}} \tag{CBN}
\end{align}
Starting from Example~(\ref{ex4}), the CBV equational theory can obtain the term
$
\lambda x y. \ureset{(\lambda a.\ureset{x a})\big((\lambda a.\ureset{x a}) y\big)}$,
and then also the term
$
\lambda x y. \ureset{(\lambda a.\ureset{x a})\ureset{x y}},
$
however, no further rewriting is possible using that theory, because neither is $\ureset{x []}$ a pure evaluation context nor is $\ureset{x y}$ a value. As for the CBN equational theory, it can not rewrite the starting term (\ref{ex4}), because there are nested applications to the continuation variable $k\hookleftarrow (k \hookleftarrow y)$ and, unlike in the CBV case, $x []$ is not a pure evaluation context in CBN. Note that: 1) there is no $x$ missing in the output of CBN TDPE; 2) the term
$
\lambda x y. \ureset{(\lambda a.\ureset{x a})\ureset{x y}}
$
obtained by CBV equations is normalized by the CBV TDPE procedure to the same thing as term (\ref{ex4}) (see \texttt{nbe\_tests.v} from the implementation), meaning that TDPE knows how to further reduce those ``blocked'' terms.

The following example is very similar to the previous one, hence we will not comment on it much, but its purpose is to show that CBN TDPE can duplicate the variable $x$ if needed.
\begin{align}
\lambda x y.& \ureset{x\ureset{x (\ushift{k}{k(k y)})}}\label{ex5}\\
\lambda x y.& {\ureset{x (x (x y))}} \tag{CBV}\\
\lambda x y.& \ureset{\ureset{x\ureset{x \ureset{y}}}} \tag{CBN}
\end{align}

The next example contains an evaluation context and a form of shift that can be used with the equational theory for CBN.
\begin{align}
\lambda x y.& \ureset{(\ushift{k}{k y}) x}\label{ex6}\\
\lambda x y.& \ureset{y x} \tag{CBV}\\
\lambda x y.& \ureset{\ureset{y\ureset{x}}} \tag{CBN}
\end{align}
The CBV equational theory rewrites (\ref{ex6}) to $\lambda xy.\ureset{\ureset{y x}}$, and the CBN one rewrites (\ref{ex6}) to $\lambda xy.\ureset{y x}$. If we apply TDPE on these results of rewriting, we get the same output as the TDPE for (\ref{ex6}).

In the following example, two shifts interact inside the same reset.
\begin{align}
\lambda x y z.& \ureset{\big(\ushift{k}{y(k z)}\big)\big(\ushift{k'}{z(k' x)}\big)}\label{ex7}\\
\lambda x y z.& \ureset{y(z(zx))} \tag{CBV}\\
\lambda x y z.& \ureset{\ureset{y\ureset{z\ureset{z\ureset{x}}}}} \tag{CBN}
\end{align}
The CBV equational theory can transform (\ref{ex7}) to $\lambda x y z.\ureset{y\ureset{z\ureset{z x}}}$ which can in turn be transformed by CBV TDPE to the same result as for (\ref{ex7}). The CBN equational theory can rewrite the left occurrence of shift and obtain $\lambda x y z. \ureset{y\ureset{z\big(\ushift{k'}{z(k' x)}\big)}}$, but no further rewriting is possible because $z[]$ is not a pure evaluation context in CBN; nevertheless, $\lambda x y z. \ureset{y\ureset{z\big(\ushift{k'}{z(k' x)}\big)}}$ can further be transformed by CBN TDPE to the same output as for (\ref{ex7}).

We consider an example that involves sum types, but briefly, since we do not have a ready made equational theory to compare the output to.
\begin{align}
\lambda x y.& \ureset{\text{case}~x~\text{of}~(\lambda z.\ushift{k}{k z} ~\|~ \lambda z.z)}\\
\lambda x y.& \text{case}~x~\text{of}~(\lambda z.\ureset{z} ~\|~ \lambda z.\ureset{z}) \tag{CBV}\\
\lambda x y.& \text{case}~x~\text{of}~(\lambda z.\ureset{\ureset{\ureset{z}}} ~\|~ \lambda z.\ureset{\ureset{\ureset{z}}}) \tag{CBN}
\end{align}
We see that not only the reset is pushed from the front of the case-expression into its branches. 

\section{Discussion and Related Work}\label{sec:conclusion}

The CPS translation that we use for CBV TDPE is exactly\footnote{There are additional typing annotations concerning worlds attached to the continuations at the type theoretic level. Another subtle point is that, when evaluating reset (figures \ref{fig:evalcbn} and \ref{fig:evalcbv}), the type theoretic model predicates when to insert a syntactic reset between the return and the run. One may insert a reset in the other cases as well, if one wants to obtain normal forms with more resets, but we prefer to not do it since it is not mandated by the model.} the standard (non-extended) CBV CPS translation of Danvy and Filinski \cite{DanvyF1990}, known also as 1-CPS. Terms in 1-CPS arising from shift are not evaluation-order independent when executed in regular functional programming languages, and that is why, to fix the semantics of shift and reset regardless of the target language of CPS, an additional CBV CPS translation of the CPS result is usually performed, and this composition of two CPS translations is known as the extended CPS, or 2-CPS. It is with respect to this 2-CPS that Kameyama \cite{Kameyama2007} proved the equational theory for CBV to be sound and complete.

The CPS translation used for CBN TDPE is \emph{not} the available 1-CPS of Biernacka and Biernacki \cite{BiernackaB2009}. The difference is in the shift rule (see Figure~\ref{fig:evalcbn}), as Biernacka and Biernacki's
$
    \eval{\shift{p}}{\rho} := \kappa\mapsto\run{\eval{p}{\kappa,\rho}}
$
would not type check in our type theoretic model. The standard 2-CPS translation, that Kameyama and Tanaka \cite{KameyamaTanaka2010} proved their equational theory for CBN sound and complete for, is obtained by performing a \emph{CBV} translation of the 1-CPS CBN translation.

We profit from our implementation language having strong reduction\footnote{Constructive type theory is strongly normalizing and there is a simple and efficient implementation of a virtual machine for strong reduction \cite{GregoireL2002}.} in that we do not have to apply two passes of CPS. That is, 1-CPS is sufficient because our evaluation (CPS translation) of a term at the meta-level is a typed and closed term which reduces to the same normal form regardless of the reduction strategy. The typed CPS-s used by Kameyama and Tanaka need recursive types, while we do not. On the other hand, we do not know if it possible at all to account for constants defined by general recursion in our model \footnote{We have however, in separate work, extended the model to higher type primitive recursion (Godel's System T) plus shift and reset on numeric types.}. The question, therefore, of whether our TDPE could be useful in practice is open. We certainly find it useful when ``practice'' concerns lambda calculi for Logic and proof assistants.

We saw in Subsection~\ref{examples} that some terms that cannot be further rewritten by the equational theories, can be further normalized by the TDPE. On the other hand, the equational theories have been proven to be sound and complete with respect to the CPS translation, that is, an equation holds between two terms if and only if the two terms have $\beta$-$\eta$-equal CPS translations. That means that the extra ``rewriting'' done by the TDPE somehow extends the equality of CPS translations -- indeed, the outputs of TDPE for the examples of Subsection~\ref{examples} do \emph{not} have CPS translations $\beta$-$\eta$-equal with the ones of the inputs. Nevertheless, at least for all the examples that we have tested, the TDPE identifies the original terms with the ``intermediary'' results arrived to by the equational theories.

As for the typing system we use, we note that it is Filinski's system \cite{FilinskiThesis}, which is sufficient for representing monadic effects by delimited control. A difference with that system is that there is an annotation $b$ on the turnstile ($\vdash_b$) whose purpose is to not allow shifts appearing outside the delimited. This could have also been guaranteed by an external syntactic criteria on whole terms, but the calculus is easier to model if all information is already present in the typing system. The more general typing system for shift and reset, with answer type modification, can type check more programs, but with the price of a function being able to modify its own answer type that we are not ready to pay. In particular, the modified meening of implication would not immediately correspond to something well known on the side of Logic.

Our work was developed independently of the results of the previous work on TDPE for CBV shift and reset of Tsushima and Asai \cite{TsushimaAsai2009}, which seems to derive from their preceding works on traditional offline and online partial evaluation for shift and reset \cite{Asai2002,Asai2004}. The difference between theirs and our results seems to be that: 1) they treat a more general typing system (function types with answer type modification); 2) we aim to produce normal forms that eliminate as many shifts and resets as possible: for example, during reification for function types, Tsushima and Asai's TDPE constructs a $\shift{\cdot}$ immediately after the first $\lam{\cdot}$, whereas we postpone the construction of $\shift{\cdot}$ to some cases of reification that will subsequently called.

\subsection*{Acknowledgements}
I thank the anonymous referees for pointing out problematic parts that led to improvement of the paper. 

\bibliographystyle{eptcs}
\bibliography{nbe-shift-1}

\end{document}